\documentclass{article}

\title{ Maximum Linear Arrangement is NP complete on Bipartite Graphs} 
\author{Aparna Das \and Jonathan Needleman}

\usepackage{amsmath}
\usepackage{amsthm}
\usepackage{amsfonts}
\usepackage{amssymb}
\usepackage{amscd}
\usepackage{calc}
\usepackage{graphicx}
\usepackage[arc, all]{xy}
\usepackage{color}
\usepackage{wasysym}
\usepackage{setspace}
\usepackage{url}
\usepackage{hyperref} 
\usepackage{comment}

\newtheorem{theorem}{Theorem}
\newtheorem{prop}[theorem]{Proposition}
\newtheorem{lemma}[theorem]{Lemma}
\newcommand{\maxLA}
{\texttt{MaxLA}}
\newcommand{\minLA}{\texttt{MinLA}}

\newcommand{\ds}{\displaystyle}
\xyoption{line}

\begin{document}

\maketitle

\begin{abstract}
The linear arrangement of a simple graph $G = (V,E)$ is a bijective labeling $\pi$ of the vertices of $G$ to the set $\{1, \ldots, n\}$.  The maximum linear arrangement ($\maxLA$) problem is to determine whether $G$ has a labeling such that the sum of labels across $G$'s edges, $\sum_{uv \in E} |\pi(u) - \pi(v)|$, is at least a given value $k$.  Minimum linear arrangement ($\minLA$) is defined analogously.

Both $\maxLA$ and $\minLA$ are known to be NP complete for general graphs.  $\minLA$ was shown to be NP complete for bipartite graphs by Evans and Shiloah in 1975~\cite{even}.  It has been an open problem whether the same holds for $\maxLA$.  In this paper we prove that $\maxLA$ is NP complete for bipartite graphs.

\end{abstract}

\section{Introduction}
A linear arrangement of a simple graph $G = (V,E)$ is a bijective labeling $\pi$ of the vertices of $G$ to the set $\{1, \ldots, n\}$.  The linear arrangement sum is defined by extending $\pi$ to a function on the edges, where for each edge $e=xy\in E$, $\pi(e)=|\pi(x)-\pi(y)|$. The sum is then
\[D_\pi=\sum_{e\in E} \pi(e)\]
and the maximum sum of $G$ is defined as
$D(G)=\max_{\pi} D_\pi$.

The maximum linear arrangement problem ($\maxLA$) takes a graph $G$ and a value $k$ and asks whether $D(G)\ge k$.  The minimum linear arrangement problem ($\minLA$) is defined analogously.

Both $\maxLA$ and $\minLA$ were shown to be NP complete on general graphs  \cite{garey}. 
Since the 1970s, $\minLA$ has been studied more extensively than $\maxLA$. For example, $\minLA$ is also known to be NP Complete for bipartite graphs~\cite{even}, interval graphs, and permutation graphs~\cite{cohen}. See \cite{diaz} and \cite{petit} for a survey. 

Recently, however there have been a number of new results for $\maxLA$. DeVos and Nurse show that the problem is NP-Hard on directed planar graphs~\cite{devosNurse}, suggesting $\maxLA$ may be NP-Hard also for bipartite graphs. In constrast, Hassin and Rubinstein showed a $2$-approximation exists for the generalized weighted version of $\maxLA$~\cite{hassinRubinstein}. Azar and Vainstein improved these results by showing that in metric space the problem admits an approximation scheme ($1-\epsilon$- approximation for any constant $\epsilon > 0$)~\cite{azarVainstein}.  
Alemany-Puig et al. showed that  $\maxLA$ can be solved for projective and planar variants of trees in linear time~\cite{Puig_Math-24}. The same authors studied a variation of $\maxLA$ for bipartite graphs $G=(V_1\cup V_2, E)$, called \texttt{bipartite} $\maxLA$ where the labels of $V_1$ are constrained to be smaller than the labels of $V_2$. They give a linear time algorithm to solve \texttt{bipartite} $\maxLA$~\cite{Puig_Math_26}. These positive results give hope that $\maxLA$ could be solvable on bipartite graphs.
 
Our main result resolves this tension by showing that $\maxLA$ is NP complete on bipartite graphs.  
Since maximizing the linear arrangement sum of any graph $G$ minimizes the sum of its complement graph $\bar{G}$~\cite{devosNurse}, our result immediately implies that $\minLA$ is NP complete on the complement of bipartite graphs. 

\section{Notation and Known Results}
Let $G=(V,E)$ be a graph with $|V|=n$ and let $\pi$ be an arrangement of its vertices. It is helpful to think of a linear arrangement $\pi$ as a mapping of vertices $V$ onto the integers on the number line from $1$ to $n$.  Then when $\pi(x)<\pi(y)$ we can say that vertex $x$ is to the left of vertex $y$ and $y$ is to the right of $x$. 
We denote the neighbors of a vertex $u\in V$ as $N(u)=\{v\in V|uv\in E\}$ and define the left (resp. right) directional degree of a vertex $u$, as its number of neighbors to the left (resp. right) in the arrangement $\pi$ as 
 \[\mathfrak{l}_\pi(u)=|\{v\in N(u)|\pi(v)<\pi(u)\} \hskip 10mm \mathfrak{r}_\pi(u)=|\{v\in N(u)|\pi(u)<\pi(v)\}\]
For a given arrangement, the {\em cut width} of a position $1\leq i\leq n$  is 
\[c_\pi(i)=|\{uv\in E|\pi(u)\leq i, \pi(v)>i\}|.\]
It will be useful to talk about the cut width of a vertex $v\in V$.  To do this, we abuse notation and say $c_\pi(v)=c_\pi(\pi(v))$.

DeVos and Nurse introduced a robust framework for understanding arrangements \cite{nurse, devosNurse}.
Although the framework was defined for directed graphs, as noted by Alemany-Puig et al. \cite{Puig_Math_26}, it also applies to undirected graphs.
We present the parts that will be most helpful in this paper.  First, they note that cut widths sum up to $D_\pi$, namely
\begin{equation} D_\pi=\sum_{i=1}^n c_\pi(i).\label{eq:sumcutwidth}\end{equation}
They define the {\em level} for a vertex $v\in V$ as the difference between its right and left degrees ie., $l_\pi(v)=\mathfrak{r}_\pi(v)   - \mathfrak{l}_\pi(v)$.  Level values are connected to cut widths as follows: For an arrangement $\pi$ which orders vertices as $(v_1, v_2, \ldots, v_n)$,
\begin{equation} c_\pi(i)=\sum_{j=1}^i l_\pi(v_i). \label{eq:sumlevels}\end{equation}
Thus $D_\pi$ can be calculated by using the cut widths or levels produced by $\pi$.  

DeVos and Nurse also show that in a {\em maximal} linear arrangement level values are monotonically increasing from left to right and vertices of equal levels can exchange positions in the arrangement without decreasing the sum.   
\begin{lemma}\label{lem:level} \cite{nurse, devosNurse}
If $\pi^*$ is a maximal linear arrangement, then the level values of vertices are monotonically non-increasing with respect to the arrangement.  Thus if $\pi^*(u)<\pi^*(v)$ then $l_{\pi^*}(u)\geq l_{\pi^*}(v)$.
Additionally, vertices with equal levels can exchange positions without decreasing $D_{\pi^*}$.
\end{lemma}   

\section{The Main Result}
To show NP completeness we reduce the minimum bisection problem to an instance of $\maxLA$ on a bipartite graph. The minimum bisection problem which is known to be NP complete \cite{garey}  takes as input a simple graph $G=(V, E)$ with an even number of vertices. The goal is to partition the vertices of $G$ into two sets of equal size, so that the number of edges with exactly one endpoint in each set is minimized.

For our reduction, we construct a new bipartite graph $G'=(V', E')$ that contains all original $2n$ vertices of $G$ along with $O(n^5)$ new vertices. We add edges to $G'$ to ensure that $D(G')$ is maximized only if the original vertices are bisected into two separate regions of equal sizes, and if the number of edges that go between the two regions is minimized.

The construction of $G'$ is described in section \ref{sec:graph}. The subsequent sections proves properties about the maximum linear arrangement of $G'$ and Theorem \ref{thm:np} uses these properties to show that the $\maxLA$ arrangement of $G'$ implies a minimum bisection for $G$. 

\subsection{Bipartite Graph Construction for MaxLP} \label{sec:graph}
Define a bipartite graph $G'=(V', E')$ from $G=(V,E)$ as follows. Let $V_E= \{v_{ij}|1\leq i< j\leq 2n, v_iv_j\in E  \}$ and $ U_i = \{u_{ij}| 1\leq j\leq n^4\}$. Define $V'$ as
\[V'= V\cup V_E \cup \left(\bigcup_{i=1}^{2n} U_i\right)\]
For convenience we will denote $U=\bigcup_{i=1}^{2n} U_i$. There will be two types of edges in $E'$ : first $v_iv_{ij}\in E'$ if and only if $v_iv_j\in E$. Second $v_iu_{jk}\in E'$ if and only if $i\neq j$, thus $v_i$ is connected to all the $n^4$ vertices in each  $U_j$ where $j\ne i$.  An illustration of $G'$ is shown in Figure \ref{fig:G'}. 
\begin{figure}[htbp!]
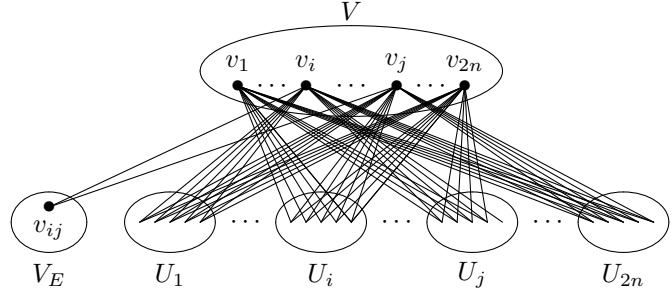

\[\xy
(0,10)*\ellipse(20,6){-};(0,28)*{V};
(-15,18)*{\bullet};(-15,21)*{v_1};
(-6,18)*{\bullet};(-6,21)*{v_i};
(6,18)*{\bullet};(6,21)*{v_j};
(15,18)*{\bullet};(15,21)*{v_{2n}};
(0,18)*{\cdots};(10.5,18)*{\cdots};(-10.5,18)*{\cdots};
(-20,0)*\ellipse(5,4){-};(-40,-7)*{V_E};
(-12,0)*\ellipse(6,4){-};(-24,-7)*{U_1};
(-2,0)*\ellipse(6,4){-};(-4,-7)*{U_i};
(8,0)*\ellipse(6,4){-};(16,-7)*{U_j};
(18,0)*\ellipse(6,4){-};(36,-7)*{U_{2n}};
(-14,0)*{\cdots};(6,0)*{\cdots};(26,0)*{\cdots};
(-15,18)*{};(-8,0)*{}**\dir{-};(-15,18)*{};(-6,0)*{}**\dir{-};(-15,18)*{};(-4,0)*{}**\dir{-};(-15,18)*{};(-2,0)*{}**\dir{-}; (-15,18)*{};(0,0)*{}**\dir{-};
(-15,18)*{};(12,0)*{}**\dir{-};(-15,18)*{};(14,0)*{}**\dir{-};(-15,18)*{};(16,0)*{}**\dir{-};(-15,18)*{};(18,0)*{}**\dir{-}; (-15,18)*{};(0,0)*{}**\dir{-};
(-15,18)*{};(32,0)*{}**\dir{-};(-15,18)*{};(34,0)*{}**\dir{-};(-15,18)*{};(36,0)*{}**\dir{-};(-15,18)*{};(38,0)*{}**\dir{-}; (-15,18)*{};(40,0)*{}**\dir{-};
(-6,18)*{};(-28,0)*{}**\dir{-};(-6,18)*{};(-26,0)*{}**\dir{-};(-6,18)*{};(-24,0)*{}**\dir{-};(-6,18)*{};(-22,0)*{}**\dir{-}; (-6,18)*{};(-20,0)*{}**\dir{-};
(-6,18)*{};(12,0)*{}**\dir{-};(-6,18)*{};(14,0)*{}**\dir{-};(-6,18)*{};(16,0)*{}**\dir{-};(-6,18)*{};(18,0)*{}**\dir{-}; (-6,18)*{};(20,0)*{}**\dir{-};
(-6,18)*{};(32,0)*{}**\dir{-};(-6,18)*{};(34,0)*{}**\dir{-};(-6,18)*{};(36,0)*{}**\dir{-};(-6,18)*{};(38,0)*{}**\dir{-}; (-6,18)*{};(40,0)*{}**\dir{-};
(6,18)*{};(-28,0)*{}**\dir{-};(6,18)*{};(-26,0)*{}**\dir{-};(6,18)*{};(-24,0)*{}**\dir{-};(6,18)*{};(-22,0)*{}**\dir{-}; (6,18)*{};(-20,0)*{}**\dir{-};
(6,18)*{};(-8,0)*{}**\dir{-};(6,18)*{};(-6,0)*{}**\dir{-};(6,18)*{};(-4,0)*{}**\dir{-};(6,18)*{};(-2,0)*{}**\dir{-}; (6,18)*{};(0,0)*{}**\dir{-};
(6,18)*{};(32,0)*{}**\dir{-};(6,18)*{};(34,0)*{}**\dir{-};(6,18)*{};(36,0)*{}**\dir{-};(6,18)*{};(38,0)*{}**\dir{-}; (6,18)*{};(40,0)*{}**\dir{-};
(15,18)*{};(-28,0)*{}**\dir{-};(15,18)*{};(-26,0)*{}**\dir{-};(15,18)*{};(-24,0)*{}**\dir{-};(15,18)*{};(-22,0)*{}**\dir{-}; (15,18)*{};(-20,0)*{}**\dir{-};
(15,18)*{};(-8,0)*{}**\dir{-};(15,18)*{};(-6,0)*{}**\dir{-};(15,18)*{};(-4,0)*{}**\dir{-};(15,18)*{};(-2,0)*{}**\dir{-}; (15,18)*{};(0,0)*{}**\dir{-};
(15,18)*{};(12,0)*{}**\dir{-};(15,18)*{};(14,0)*{}**\dir{-};(15,18)*{};(16,0)*{}**\dir{-};(15,18)*{};(18,0)*{}**\dir{-}; (15,18)*{};(0,0)*{}**\dir{-};
(-40,2)*{\bullet};(-40,-1)*{v_{ij}};
 (-40,2)*{};(-6,18)*{}**\dir{-}; (-40,2)*{};(6,18)*{}**\dir{-};
\endxy\]
\caption{Blowup of $G$ to construct $G'$}\label{fig:G'}
\end{figure}

\subsection{Arranging original vertices}
 We establish how a $\maxLA$ labeling of $G'$ will arrange the original vertices $V$ relative to $U$. By Lemma \ref{prop:u clump} each $v\in V$ will be either left of, or right of all of $U$.

\begin{prop}\label{prop:u clump}
	Let $\pi^*$ be a maximal labeling of $G'$. Then for each $v_i\in V,$ $v_i$ is either left of all $u\in U$ or right of all $u\in U$. 
\end{prop}
\begin{proof}
	 Let $\pi$ be an arrangement of $G'$ and let $v_i \in V$ be the leftmost  vertex  with $U$ vertices on both its left and right. Define those left and right vertices as
     \[L_i=\{u\in U|\pi(u)<\pi(v_i) \} \textrm{~~and~~} R_i=\{u\in U|\pi(u)>\pi(v_i) \}.\]
Without loss of generality assume $|L_i|\leq |R_i|$ and let $u_L$ be the leftmost vertex of $L_i$. Define a new function $\pi'$ on $G'$ which agrees with $\pi$ everywhere except it swaps the labels of $v_i$ and $u_L$ so that, $\pi'(v_i)=\pi(u_L)$ and $\pi'(u_L)=\pi(v_i)$.
	
    
We  show $D_{\pi'}\geq D_\pi$. As only edges that involve $v_i$ or $u_L$ change from $D_{\pi'}$ and $D_\pi$, we show that
\[D_{\pi'}-D_\pi=\sum_{x\in V'} \pi'(xu_L)-\pi(xu_L)+\sum_{x\in V'} \pi'(xv_l)-\pi(xv_i) >0.\]

This sum is dependent on $d=|\pi(u_L)-\pi(v_i)|$.  Consider the first sum.   Swapping the label of $u_L$ with $v_i$ may have moved $u_L$ closer to each of its at most $2n-1$ connections in $V$ by $d$. Therefore, \[\sum_{x\in V'} \pi'(xu_L)-\pi(xu_L) \geq -d(2n-1).\]
The second sum is more complicated to analyze as there are more edges involving $v_i$. We break up the sum by different sets of edges as 
\[\sum_{x\in V'} \pi'(xv_i)-\pi(xv_i)=\sum_{u\in R_i} \pi'(uv_i)-\pi(uv_i)+\]
\[\sum_{u\in L_i} \pi'(uv_i)-\pi(uv_i)+\sum_{v_{ij}\in V_E} \pi'(v_{ij}v_l)-\pi(v_{ij}v_l).\]
There are at least $|R_i| - n^4$ edges from $R_i$ to $v_i$. Each of these edges contribute $d$ more to the sum under $\pi'$ thus,
\[\sum_{uv_i\in R_i} \pi'(uv_i)-\pi(uv_i)\geq d(|R_i|-n^4).\] 
Next, consider the edges from $L_i$. Let $X=\{x\in V'| \pi(u_L)\leq \pi(x)<\pi(v_i)\}$ be the vertices between $u_L$ and $v_i$.  Notice, that $L_i\subseteq X$ and that $|X|=d$ by the choice of $u_L$ as the leftmost $U$ vertex from $v_i$. Let $a = |L_i - U_i|$ so $a$ is the number of vertices in $L_i$ with an edge to $v_i$.  Let $b = d-a$, be the number of the vertices from $X$ without edges to $v_i$. The sum $\ds \sum_{u\in L_i-U_i} \pi'(uv_i)-\pi(uv_i)$ is minimized when the $a$ vertices with edges to $v_i$ are leftmost in $X$.  In this case,
\[\sum_{u\in L_i} \pi'(uv_i)= \sum_{r=1}^a r \hspace{.5cm} \textrm{and} \hspace{.5cm} \sum_{u\in L_i} \pi(uv_i)= \sum_{r=1}^a b+r.\]
Therefore \[\sum_{u\in L_i} \pi'(uv_i)-\pi(uv_i)\geq -ab \ge |L_i|(n^4 + n^2)\]
where the second inequality follows as $a \le |L_i|$ and $b\le |U_i| + |V_E| \le n^4 + n^2$.

Finally, consider the edges at most $2n-1$ edges between $V_E$ and $v_i$. The contribution of each edge may have decreased by at most $d$. Therefore,
\[\sum_{v_{ij}\in V_E} \pi'(v_{ij}v_i)-\pi(v_{ij}v_i)\geq -d(2n-1).\]
Putting this all together 
\[D_{\pi'}-D_\pi \geq -d(2n-1)+d(|R_i|-n^4)-|L_i|(n^4 + n^2 )-d(2n-1)\]
As $d \ge |L_i|$, this simplifies to,
$D_{\pi'}-D_\pi \geq  d(|R_i|-2n^4 -n^2 -4n + 2).$
As $|R_i| \ge n^5$, the equation is positive for $n\ge 3$. 
\end{proof}

\subsection{Arranging U vertices}
By proposition \ref{prop:u clump} in a maximal labeling $\pi$ each $v\in V$ lies to the left of all $u\in U$  or to the right of all $u\in U$. We define $V_L$ and $V_R$ for each possibility as  
\[V_L(\pi)=\{v\in V|\pi(v)<\pi(u) \text{ for all } u\in U\}\]
\[V_R(\pi)=\{v\in V|\pi(v)>\pi(u) \text{ for all } u\in U\}.\]
Next. define $U_L, U_R$ as the $U$ vertices corresponding to $V_L$ and $V_R$ as
\[U_L(\pi)=\bigcup_{i\in \{j|v_j \in V_L\}} U_i 
\textrm{~~~and~~~} U_R(\pi)=\bigcup_{i\in \{v_j \in V_R\}} U_i .\]
Without loss of generality, denote the vertices of $V_L(\pi)$ as $v_1, v_2, \ldots, v_{k-1}$ where
$\pi(v_1) >\pi(v_2) > \cdots > \pi(v_{k-1})$,  and denote the vertices of $V_R(\pi)$ as $v_{k}, \ldots, v_{2n}$ where 
$\pi(v_{k})>\pi(v_{k+1})>\cdots>\pi(v_{2n}).$
The next lemma fixes an ordering of $U$ based on this notation and establishes that $U$ will be ordered left to right as $U_1, U_2, \ldots,U_{2n}$. Figure \ref{fig:u neat buckets} shows the arrangement.

\begin{figure}[htbp!]
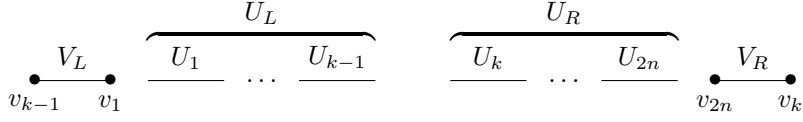

\[\xy
(-50,0)*{};(-40,0)*{}**\dir{-};(-45,3)*{V_L};
(-50,0)*{\bullet};(-50,-3)*{v_{k-1}};(-40,0)*{\bullet};(-40,-3)*{v_1};
(-35,0)*{};(-25,0)*{}**\dir{-};(-30,3)*{U_1};(-15,0)*{};(-5,0)*{}**\dir{-};(-10,3)*{U_{k-1}};
(-20,0)*{\ldots};
(50,0)*{};(40,0)*{}**\dir{-};(45,3)*{V_R};
(50,0)*{\bullet};(50,-3)*{v_{k}};(40,0)*{\bullet};(40,-3)*{v_{2n}};
(20,0)*{\ldots};
(35,0)*{};(25,0)*{}**\dir{-};(30,3)*{U_{2n}};(15,0)*{};(5,0)*{}**\dir{-};(10,3)*{U_{k}}; 
(-32,6)*{};(-5,6)*{}**\frm{^)};(-20,9)*{U_L};
(32,6)*{};(5,6)*{}**\frm{^)};(20,9)*{U_R};
\endxy\]
\caption{Arrangement of vertices after Lemma \ref{lemma:u arrangement}}\label{fig:u neat buckets}
\end{figure}


\begin{lemma}\label{lemma:u arrangement}
    If $\pi^*$ is a maximal linear arrangement of $G'$ then for all $u_l\in U_L$ and $u_r\in U_R$, $\pi^*(u_l)<\pi^*(u_r)$. Additionally, for $u_i\in U_i, u_j\in U_j$ if $i<j$ then $\pi^*(u_i)<\pi^*(u_j)$. 
\end{lemma}
\begin{proof}
The level of $u_l\in U_L$ is $l_{\pi^*}(u)=\mathfrak{r}_{\pi*}(u_l)-\mathfrak{l}_{\pi*}(u_l)=|V_R|-(|V_L|-1)$. Similarly, the level of $u_r\in U_R$ is $l_{\pi^*}(u_r)=|V_R|-1-|V_L|$.  Since $l_{\pi^*}(u_l)>l_{\pi^*}(u_r)$, by Lemma \ref{lem:level}, $\pi^*(u_l)<\pi^*(u_r)$. Furthermore, as vertices in $U_L(\pi^*)$ have the same level they can be rearranged among themselves so that $\pi^*(u_i) <\pi^*(u_j)$ when $i<j$ without decreasing $D_{\pi^*}$. The same holds for $U_R(\pi^*)$.
\end{proof}

\subsection{Bisection of original vertices}
We now show that the maximal arrangement will bisect the vertices of $G$.
\begin{lemma}\label{prop:vsplit}
Let $\pi$ be a labeling of $G'$ as in Lemma \ref{lemma:u arrangement}, then there is a labeling $\pi'$ of $G'$ such that $|V_L(\pi')| = |V_R(\pi')|$ and $D_{\pi'}\geq D_\pi$.
\end{lemma}
\begin{proof}
    Let $\pi$ be the labelling established in Lemma \ref{lemma:u arrangement}. Assume without loss of generality that $|V_R(\pi)|>|V_L(\pi)|$ and let $v_k$ be the rightmost vertex in $V_R(\pi)$.  We define a new labeling $\pi'$ which moves $v_k$ into $V_L$ as follows. Shift vertices that are  left of $v_k$ right by one by setting $\pi'(x)=\pi(x)+1$ for all $x\in V'$ where $\pi(x) < \pi(v_k)$. Then set 
    $\pi'(v_k)=1$. For all other $v\in V'$ let $\pi'(v)=\pi(v)$.  
    
    Since the sum of cut widths across all vertices is equal to the linear arrangement sum, we show $D_{\pi'}>D_\pi$ by examining how cut widths change from $\pi$ and $\pi'$.
    For convenience we define the function $\Delta c(v)=c_{\pi'}(v)-c_\pi(v)$ for $v\in V'$ and for $A\subset V'$ we say
    $\Delta c(A)=\sum_{v\in A} \Delta c(v)$. In particular, $\Delta c(V')=D_{\pi'}-D_\pi$ which we wish to show is positive.

    The only edges that change their contribution from $\pi$ to $\pi'$ are those that involve $v_k$. Thus, for any $v\in V'$ $\Delta c(v)$ depends on $v$'s position relative to $N(v_k)$. An edge $xv_k$ where $x\in N(v_k)$ counts in $c_\pi(v)$ if $\pi(x) \le \pi(v)$, and it counts in $c_{\pi'}(v)$ if $\pi'(x) > \pi'(v)$. To simplify calculations, we define for each $v\in V'$ the cut-width contribution from $v_k$ neighbors from $V_E$ and from $U$ separately as
   \begin{eqnarray*}
   c^{E}_{\pi'}(v)&=&|\{x\in V_E \mid xv_k\in E',  \pi'(x)>\pi'(v)\}|\\
   c^{E}_\pi(v)&=&|\{x\in V_E\mid xv_k\in E', \pi(x)\leq\pi(v)\}| \\
   c^{U}_{\pi'}(v)&=&|\{x\in U \mid xv_k\in E', \pi'(x)>\pi'(v)\}|  \\
   c^{U}_\pi(v)&=&|\{x\in U\mid xv_k\in E', \pi(x)\leq\pi(v)\}|
   \end{eqnarray*}
  If we define  $\Delta c^{E}(v)=c^{E}_{\pi'}(v)-c^{E}_{\pi}(v)$ and $\Delta c^{U}(v)=c^{U}_{\pi'}(v)-c^{U}_{\pi}(v)$  then $\Delta c(v)= \Delta c^{U}(v) +\Delta c^{E}(v)$. Extending $\Delta c^{E}$, $\Delta c^{U}$ to sets the way we did with $c_\pi$ we have $\Delta c(V') = \Delta c^{E}(V') + \Delta c^{U}(V')$.
  
  For any $v\in V'$, $\Delta c^{E}(v)\geq -2n$. Summing over the $2n^5 + 2n^2 +2n$ elements of $V'$ we see that $\Delta c^{E}(V')$ decreases by at most 
\begin{equation}\label{eq:c^EV'} |\Delta c^E(V')| \le |-2n|V'|| = 2n(2n^5 + 2n^2 +2n) =  O(n^6) 
\end{equation}
  
  We compute $\Delta c^{U}(V')$ in parts as 
  $\Delta c^{U}(V') = \Delta c^U(V) + \Delta c^U(V_E) + \Delta c^U(U).$ 
  For each $v\in V_L \cup\{v_k\}$, as $v_k$ was moved to the leftmost position of $V_L$ by $\pi'$, $\Delta c^{U}(v)=|U-U_k|= 2n^5-n^4$.  For $v\in V_R-\{v_k\}$, as $v_k$ was in rightmost position of $V_R$ in $\pi$, $\Delta c^{U}(v) = -|U-U_k|= -(2n^5-n^4)$. Combining, we have $\Delta c^{U}(V) = (2n^5-n^4)(|V_L|+1)-(2n^5-n^4)(|V_R|-1) = (2n^5-n^4)(|V_L|-|V_R|+2)$. 
As $|V_L| - |V_R| \ge -2n$, $\Delta c^{U}(V)$ decreases by at most  
     \begin{equation} |\Delta c^{U}(V)| \le |(2n^5-n^4)(-2n+2)| = O(n^6) 
     \label{eq:deltac^uV} \end{equation}

Next, for $v \in V_E$, $\Delta c^{U}(v)\geq -|U-U_k|=-(2n^5-n^4)$. As $|V_E| \le 2n^2$, summing over $V_E$ shows that $\Delta c^{U}(V_E)$ decreases by at most
\begin{equation} |\Delta c^{U}(V_E)|\leq (2n^5-n^4)|V_E| = O(n^7)
\label{eq:deltac^uV_E}\end{equation} 
    
Turning to $\Delta c^{U}(U)$, for  $1\leq i \leq k-1$,  $\Delta c^{U}(U_i)=-\Delta c^{U}(U_{2n-i+1})$.  Similarly, for $k+1\leq i\leq n$,  $\Delta c^{U}(U_i)=-\Delta c^{U}(U_{2n-i+2})$. As a result of these cancellations, $\Delta c^{U}(U)=\Delta c^{U}(U_k)+\Delta c^{U}(U_{n+1})$. 
We look at $\Delta c^U(U_k)$. Under $\pi'$ the cut width of $u\in U_k$ includes all vertices in $\bigcup_{i=k+1}^{2n} U_i$, so $c^U_{\pi'}(u)= n^4(2n-k)$. Under $\pi$ the cut width includes all vertices in $\bigcup_{i=1}^{k} U_i$, so $c^U_\pi(u) = n^4(k-1)$. Thus $\Delta c^U(u) = n^4(2n-2k+1)$. Summing over the $n^4$ vertices of $U_k$ yields $\Delta c^{U}(U_k)=n^8(2n-2k+1).$ As $k<n, \Delta c^{U}(U_k)\ge n^8$.
   
   Finally we calculate $\Delta c^{U}(U_{n+1})$.  As $v_k$ is adjacent to $n^5-n^4$ vertices in $\bigcup_{i=1}^n U_i$ and also in $\bigcup_{i=n+2}^{2n}U_i$, those edges cancel in the computation of $\Delta c^U(U_{n+1})$. Now consider the edges from $v_k$ to $U_{n+1}$ by numbering the vertices of $U_{n+1}$ from left to right as $u_1, u_2, \ldots, u_{n^4}$. 
   For $u_j\in U_{n+1}$, there are $n^4-j$ vertices of $U_{n+1}$ that are right of $u_j$ and $j$ that are left of $u_j$ which are adjacent to $v_k$. So $\Delta c^{U}(u_j)=n^4-2j$. Summing over $1\le j \le n^4$, yields $\Delta c^{U}(U_{n+1}) =-n^4.$ Thus $\Delta c^U(U)$ increases by at least
   $$\Delta c^U(U) = \Delta c^U(U_k) + \Delta c^U(U_{n+1}) \ge n^8 - n^4 = n^8+O(n^4)$$ 
   Combining the above with equations \ref{eq:deltac^uV}, \ref{eq:deltac^uV_E} and  \ref{eq:c^EV'}, we find that $\Delta c(V')$ is positive for sufficiently large $n$ as $n^8$ is the dominant term below
    \begin{eqnarray*}
    \Delta c(V') & = &  \Delta c^E(V') + \Delta c^U(V) + \Delta c^U(V_E) + c^U(U)\\
    &=& O(n^6)+O(n^6)+O(n^7) + n^8 + O(n^4)  \ge 0.
    \end{eqnarray*} 
\end{proof}

\subsection{Arranging edge vertices}
We analyze the arrangement of $V_E$ considering three groups: 
    $V_{LL}=\{v_{ij}| v_i, v_j \in V_L, 1\leq i<j \leq n\}$, both edge endpoints are in $V_L$, 
    
    $V_{RR}=\{v_{ij}| v_i, v_j \in V_R, n+1 \le i < j \le 2n\}$, both endpoints are in $V_R$,
    
    $V_{LR}=V_E-\left(V_{LL}\cup V_{RR}\right)$, one endpoint in $V_L$ and the other in $V_R$. 
    
    \noindent The next lemma shows that $V_{RR}$ will be between $V_L$ and $U_L$, $V_{LR}$ will be between $U_L$ and $U_R$ and $V_{LL}$ will be between $U_R$ and $V_R$. The final arrangement is shown in Figure \ref{fig:final_arrange}.
\begin{lemma}\label{lemma:edgearrange}
Let $\pi^*$ be a maximal labeling of $G'$ that conforms to Lemma \ref{prop:vsplit}.

    If $v_{ij}\in V_{RR}$,  then  for all $u_l\in U_L$, and $v_l\in V_L$,  
    $\pi^*(v_l)< \pi^*(v_{ij}) < \pi^*(u_l)$.  
    
    If $v_{ij}\in V_{LL}$, then for all $u_r \in U_R$, and $v_r\in V_R$, $\pi^*(u_r) <\pi^*(v_{ij}) < \pi^*(v_r)$. 
    
    If $v_{ij} \in V_{LR}$, then for all $u_l \in U_L$ and $u_r \in U_R$,  $\pi^*(u_l) < \pi^*(v_{ij})< \pi^*(u_r)$. 
\end{lemma}
\begin{proof}
Observe that as all $U$ vertices appear to the left of $v_1$, the level of $v_l\in V_L$ is at least $(2n-1)n^4-(2n-1) >n^4$. Similarly the level of $v_r\in V_R$ is at most $(2n-1)-(2n-1)n^4 <-n^4$.
Next, observe that the level of $v_{ij} \in V_E$ must be one of the following values $\{-2, 0, 2\}$ as each $v_{ij}$ has exactly two neighbors. Thus $l(v_l)> n^4 > l(v_{ij}) > l(v_r) >-n^4$ which means that a $v_{ij} \in V_E$ must lie between $V_L$ and $V_R$. So we have $\pi^*(v_l) <\pi^*(v_{ij}) < \pi^*(v_r)$.

Now consider the level of $U$ vertices. Since by Lemma \ref{prop:vsplit} $|V_L| = |V_R|$, the level of $u_l\in U_L$ is $|V_R| - (|V_L|-1) = 1$, and the level of $u_r\in U_R$ is $|V_R|-1 - |V_L| = -1$.

The two neighbors of $v_{ij} \in V_{RR}$, are both in $V_R$, so $l(v_{ij}) = 2 > l(u_l) = 1$. Thus, $v_{ij}$ lies to the left of $u_l \in U_L$ so $\pi^*(v_{ij}) < \pi^*(u_l)$.

For $v_{ij} \in V_{LL}$, as its two neighbors are both in $V_L$, $l(v_{ij}) = -2 < l(u_r) = -1$. Thus, $v_{ij}$ must lie to the right of $u_r\in U_R$ so $\pi^*(u_r) < \pi^*(v_{ij})$.

Finally $v_{ij} \in V_{LR}$, has one neighbor in each $V_L$ and $V_R$, so $l(v_{ij}) = 0$. Thus $v_{ij}$ must lie between $u_l\in U_L$ and $u_r \in U_R$ so $\pi^*(u_l) < \pi^*(v_{ij}) < \pi^*(u_r)$.
\end{proof}

\subsection{Reduction to minimum bisection}

\begin{theorem}\label{thm:np}
$\maxLA$ is NP complete for bipartite graphs.
\end{theorem}
\begin{proof}
 Given input $G=(V,E)$ for the minimum bisection problem we construct $G'$ as described in section \ref{sec:graph}. A $\maxLA$ arrangement of $G'$, $\pi$, satisfies the conditions of Lemma \ref{lemma:edgearrange}. Figure \ref{fig:final_arrange} shows how $V'$ is arranged by $\pi$ and denotes the number of vertices in each region as $m=|V_E|$ and $m_L=|V_{LL}|, m_R = |V_{RR}|$, and $m_C=|V_{LR}|$.  By Lemma \ref{prop:vsplit}, the original vertices of $G$ are bisected and $|V_L|=|V_R|=n$. 
To finish the proof we show that $D_{\pi}$ is maximized when the number of edges between the bisected vertices, $m_C$, is minimized.  Thus solving $\maxLA$ on $G'$ gives us a minimum bisection solution for $G$.

\begin{figure}[h!]
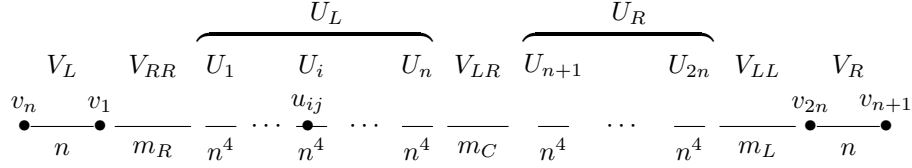

\[\xy
(-60,0)*{\bullet};(-50,0)*{\bullet}**\dir{-};(-55,8)*{V_L};(-55,-3)*{n};(-60,3)*{v_n};(-50,3)*{v_1};
(-48,0)*{};(-38,0)*{}**\dir{-};(-43,8)*{V_{RR}};(-43,-3)*{m_R};
(-36,0)*{};(-32,0)*{}**\dir{-};(-34,8)*{U_1};(-28,0)*{\cdots};(-34,-3)*{n^4};
(-24,0)*{};(-20,0)*{}**\dir{-};(-22,8)*{U_i};(-15,0)*{\cdots};(-22,-3)*{n^4};
(-34,12)*{};(-6,12)*{}**\frm{^)};(-20,15)*{U_L};
(-22.5,0)*{\bullet};(-22.5,3)*{u_{ij}};
(-10,0)*{};(-6,0)*{}**\dir{-};(-8,8)*{U_n};(-8,-3)*{n^4};
(-4,0)*{};(4,0)*{}**\dir{-};(0,8)*{V_{LR}};(0,-3)*{m_C};
(54,0)*{\bullet};(44,0)*{\bullet}**\dir{-};(49,8)*{V_R};(49,-3)*{n};(44,3)*{v_{2n}};(54,3)*{v_{n+1}};
(42,0)*{};(32,0)*{}**\dir{-};(37,8)*{V_{LL}};(37,-3)*{m_L};
(30,0)*{};(26,0)*{}**\dir{-};(28,8)*{U_{2n}};(19,0)*{\cdots};(28,-3)*{n^4};
(12,0)*{};(8,0)*{}**\dir{-};(10,8)*{U_{n+1}};(10,-3)*{n^4};
(28,12)*{};(6,12)*{}**\frm{^)};(20,15)*{U_R};
\endxy\]
\caption{Arrangement of vertices as in Lemma \ref{lemma:edgearrange}}\label{fig:final_arrange}
\end{figure}
To calculate $D_{\pi}$ we sum cut widths, which are themselves found by summing level values of a vertex and those to its left using equations \ref{eq:sumcutwidth} and \ref{eq:sumlevels}.

We calculate cut widths for each set of vertices separately starting with $v_i\in V_L$. For $v_i\in V_L$, $l_\pi(v_i)=2n^5-n^4+\text{deg}(v_i)$. As $0\le\text{deg}(v_i)\le 2n-1$, by equation \ref{eq:sumlevels} the cut width of $v_i$ is
\[c_\pi(v_i) = \sum_{j=1}^{n-i+1}l_\pi(v_{n-j+1}) = \sum_{j=1}^{n-i+1}2n^5 - n^4  + O(n).\] 
Summing across the $n$ vertices in $V_L$, we find 
\[c_\pi(V_L)  =  \sum_{i=1}^n(n-i+1)(2n^5-n^4 +O(n))  =  \frac{n(n+1)}{2}(2n^5-n^4+O(n))\] 
Next consider $v_i \in V_R$. We sum the levels of all vertices to its left. The sum of the levels of $V_L$ is given by $c_\pi(v_1)$. For edge vertices $e\in V_{RR}$, $l_\pi(e)=2$, for $e\in V_{LR}$, $l_\pi(e)=0$ and for $e\in V_{LL}$,  $l_\pi(e)=-2$. For $u\in U_L$, $l_\pi(u)=1$, and for $u\in U_R$, $l_\pi(u)=-1$.  Finally, for $v_i \in V_R$, $l_\pi(v_i) = -2n^5-n^4-\text{deg}(v_i) = -2n^5-n^4 + O(n)$. Thus the cut width for $v_i \in V_R$ where $n+1\leq i \leq 2n$ is 
 \[ c_\pi(v_i) = c_\pi(v_1)+2(m_R-m_L) -(2n-i+1)(2n^5-n^4+O(n))
 \]
Summing $c_\pi(v_i)$ over the vertices in $V_R$ for $n+1\leq i \leq 2n$ yields 
\[c_\pi(V_R) \ge 
nc_\pi(v_1)+2n(m_R-m_L)-\frac{n(n+1)}{2}(2n^5-n^4+O(n))\]
Calculating $c_\pi(v_1)$ directly--there are $n(2n^5-n^4)$ edges from $V_L$ to $U$, $2m_L$ edges to $V_{LL}$, and $m_C$ edges to $V_{LR}$. Thus,
\[c_\pi(v_1)=n(2n^5-n^4)+2m_L+m_C
  = 2n^6 - n^5 + 2m_L + m_C\]
Summing $c_\pi(V_L)$ and $c_\pi(V_R)$ and substituting in the value of $c_\pi(v_1)$ we obtain
\[c_\pi(V) = c_\pi(V_L) + c_\pi(V_R) =  n(2n^6-n^5+2m_R + m_C) = 2n^7-n^6+n(2m_R+m_C)\]
As $O(m)\leq O(n^2)$, the sum of $V$'s cut widths is $c_\pi(V)=2n^7-n^6+O(n^3)$. 

Next we look at the cut widths sums of $V_{RR}, V_{LL}$ and $V_{LR}$.  Label the vertices of each set from left to right as $e_1, e_2, \ldots$. For $e_i \in V_{RR}$ $c_\pi(e_i)=c_\pi(v_1)+2i$. For $e_i \in V_{LR}$, $c_{\pi}(e_i) = c_\pi(v_1)+2m_R+n^5$. For $e_i \in V_{LL}$, $c_\pi(e_i) = c_\pi(v_1) + 2m_R -2i$. Summing cut widths over the size of each set we have,
\begin{eqnarray*}
c_\pi(V_{RR}) &=& m_Rc_\pi(v_1)+m_R(m_R+1)\\ 
c_\pi(V_{LR})&=&m_Cc_\pi(v_1)+2m_Rm_C+m_C n^5\\
c_\pi(V_{LL})&=&m_Lc_\pi(v_1)+2m_Rm_L-m_L(m_L+1)
\end{eqnarray*}
We sum the three parts above and use that $m=O(n^2)$ to find that  
\begin{eqnarray*}
    c_\pi(V_E) &=& mc_\pi(v_1)+2m_R(m_C+m_L)+m_Cn^5\\
    & &+m_R(m_R+1)-m_L(m_L+1)\\
    &=&2mn^6-mn^5+m_Cn^5+O(n^4) 
\end{eqnarray*}

Next we look at $U_{LL}$ and $U_{RR}$. Labeling $u_{ij}\in U_L$ from left to right with $1\le i\le n$ and $1\le j\le n^4$,  $c_\pi(u_{ij})=c_\pi(v_1)+2m_R+n^4(i-1)+j$. Thus,
\begin{eqnarray*}
c_\pi(U_L)&=&\sum_{i=1}^n\sum_{j=1}^{n^4}c_\pi(v_1)+2m_R+n^4(i-1)+j\\&=&n^5c_\pi(v_1)+2n^5m_R+\frac{n(n-1)}{2}n^8+n\frac{n^4(n^4+1)}{2}.\end{eqnarray*}
Similarly, for $u_{ij}\in U_R$ with $n+1\leq i\leq 2n$, and $1 \le j \le n^4$, 
 $c_\pi(u_{ij})=c_\pi(v_1)+2m_R+n^5-n^4(i-n-1)-j$. Thus,
\[c_\pi(U_R)=n^5c_\pi(v_1)+2n^5m_R+n^{10}-\frac{(n-1)n}{2}n^8-n\frac{n^4(n^4+1)}{2}.\]
Summing $c_\pi(U_L)$ and $c_\pi(U_R)$ and substituting the value of $c_\pi(v_1)$ we have
\begin{eqnarray*}
    c_\pi(U)&=&2n^5c_\pi(v_1)+4n^5m_R+n^{10}\\
    &=& 2n^5(2n^6-n^5+2m_L+m_C)+4n^5m_R+n^{10}\\
    &=&4n^{11}-n^{10}+4n^5m_L+2n^5m_C+4n^5m_R\\
     &=&4n^{11}-n^{10}+4mn^5-2n^5m_C
\end{eqnarray*}
Putting the above together we calculate $D_{\pi}$ to be 
\begin{eqnarray*}
    D_{\pi} &=& c_\pi(V)+c_\pi(V_E)+c_\pi(U)\\
    &=& 2n^7-n^6+O(n^4)\\
    & &+2mn^6-mn^5+m_Cn^5+O(n^4)\\
    & & +4n^{11}-n^{10}+4mn^5-2m_Cn^5\\
    &=&4n^{11}-n^{10}+2n^7-n^6+2mn^6+3mn^5-m_Cn^5+O(n^4).
\end{eqnarray*}   
As $n$ and $m$ are fixed by the input $G$, and the term involving $m_c$ dominates $O(n^4)$, $D_{\pi}$ is maximized when $m_C$ is minimized.  Since minimizing $m_C$ (minimum bisection) is NP-hard \cite{garey}, so is maximizing $D_{\pi}$.
\end{proof}
\noindent \textbf{AI Disclosure:} We did not use any AI tools in the creation of this manuscript.
\bibliographystyle{plainurl}
\bibliography{refs.bib}

@InProceedings{azarVainstein,
  author =	{Azar, Yossi and Vainstein, Danny},
  title =	{{Multi Layer Peeling for Linear Arrangement and Hierarchical Clustering}},
  booktitle =	{50th International Colloquium on Automata, Languages, and Programming (ICALP 2023)},
  pages =	{13:1--13:18},
  series =	{Leibniz International Proceedings in Informatics (LIPIcs)},
  ISBN =	{978-3-95977-278-5},
  ISSN =	{1868-8969},
  year =	{2023},
  volume =	{261},
  editor =	{Etessami, Kousha and Feige, Uriel and Puppis, Gabriele},
  publisher =	{Schloss Dagstuhl -- Leibniz-Zentrum f{\"u}r Informatik},
  address =	{Dagstuhl, Germany},
  URL =		{https://drops.dagstuhl.de/entities/document/10.4230/LIPIcs.ICALP.2023.13},
  URN =		{urn:nbn:de:0030-drops-180656},
  doi =		{10.4230/LIPIcs.ICALP.2023.13}
}

@InProceedings{cohen,
author="Cohen, Johanne
and Fomin, Fedor
and Heggernes, Pinar
and Kratsch, Dieter
and Kucherov, Gregory",
editor="Kr{\'a}lovi{\v{c}}, Rastislav
and Urzyczyn, Pawe{\l}",
title="Optimal Linear Arrangement of Interval Graphs",
booktitle="Mathematical Foundations of Computer Science 2006",
year="2006",
publisher="Springer Berlin Heidelberg",
address="Berlin, Heidelberg",
pages="267--279",
isbn="978-3-540-37793-1"
}

@article{devosNurse,
  author       = {Matt DeVos and
                  Kathryn Nurse},
  title        = {A Maximum Linear Arrangement Problem on Directed Graphs},
  journal      = {CoRR},
  volume       = {abs/1810.12277},
  year         = {2018},
  url          = {http://arxiv.org/abs/1810.12277},
  eprinttype    = {arXiv},
  eprint       = {1810.12277},
  bibsource    = {dblp computer science bibliography, https://dblp.org}
}

@article{nurse,
  author       = {Kathryn Nurse},
  title        = {Maximum Linear Arrangement of Directed Graphs},
  journal      = {Masters Thesis, Simon Fraser Univeristy},
  year         = {2019},
  url          = {https://summit.sfu.ca/item/19406},
}

@article{hassinRubinstein,
title = {Approximation algorithms for maximum linear arrangement},
journal = {Information Processing Letters},
volume = {80},
number = {4},
pages = {171-177},
year = {2001},
issn = {0020-0190},
doi = {https://doi.org/10.1016/S0020-0190(01)00159-4},
url = {https://www.sciencedirect.com/science/article/pii/S0020019001001594},
author = {Refael Hassin and Shlomi Rubinstein}
}

@article{diaz,
author = {D\'{\i}az, Josep and Petit, Jordi and Serna, Maria},
title = {A survey of graph layout problems},
year = {2002},
issue_date = {September 2002},
publisher = {Association for Computing Machinery},
address = {New York, NY, USA},
volume = {34},
number = {3},
issn = {0360-0300},
url = {https://doi.org/10.1145/568522.568523},
doi = {10.1145/568522.568523},
journal = {ACM Comput. Surv.},
month = {sep},
pages = {313–356},
numpages = {44}
}

@article{even,
title = {NP-completeness of several arrangement problems},
journal = {Technical report TR-43, , Dept. of Computer Science, Technion, Haifa},
year = {1975},
author = {S. Even and Y. Shiloah}
}

@inproceedings{garey,
author = {Garey, M. R. and Johnson, D. S. and Stockmeyer, L.},
title = {Some simplified NP-complete problems},
year = {1974},
isbn = {9781450374231},
publisher = {Association for Computing Machinery},
address = {New York, NY, USA},
url = {https://doi.org/10.1145/800119.803884},
doi = {10.1145/800119.803884},
booktitle = {Proceedings of the Sixth Annual ACM Symposium on Theory of Computing},
pages = {47–63},
numpages = {17},
location = {Seattle, Washington, USA},
series = {STOC '74}
}

@article{petit,
author = {Petit, Jordi},
year = {2011},
month = {01},
pages = {},
title = {Addenda to the survey of layout problems},
volume = {105},
journal = {Bulletin of the European Association for Theoretical Computer Science EATCS}
}

@article{Puig_Math-24,
	title = {{The maximum linear arrangement problem for trees under projectivity and planarity}},
	journal = {Information Processing Letters},
	author = {{Alemany-Puig}, Llu\'{i}s and Esteban, {Juan Luis} and {Ferrer-i-Cancho}, Ramon},
	url = {https://arxiv.org/abs/2206.06924},
	volume = {183},
	pages = {106400},
	year = {2024},
	issn = {0020-0190},
	doi = {10.1016/j.ipl.2023.106400}
}

@article{Puig_Math_26,
	title = {{Maximum linear arrangement: exact algorithms for specific classes of graphs and approximation algorithms for wide classes of graphs}},
	author = {{Alemany-Puig}, Llu\'{i}s and Esteban, {Juan Luis} and {Ferrer-i-Cancho}, Ramon},
	volume = {52},
	ISSN = {1573-2886},
	DOI = {10.1007/s10878-026-01454-z},
	number = {16},
	journal = {Journal of Combinatorial Optimization},
	publisher = {Springer Science and Business Media LLC},
	year = {2026},
	month = {09}
}

\end{document}